\pdfoutput=1
\documentclass[journal]{IEEEtran}

\usepackage[hyphens]{url}

\usepackage{graphicx}
\usepackage{caption}
\usepackage{subcaption}
\usepackage{amsmath}
\usepackage{amsthm}

\usepackage[table]{xcolor}
\usepackage{mathtools}
\usepackage{amssymb}
\usepackage{tabulary}
\usepackage{booktabs}

\usepackage{algorithmic}
\usepackage[ruled, vlined, linesnumbered]{algorithm2e}

\usepackage{enumitem}
\usepackage{multirow}
\usepackage{stfloats}
\usepackage{threeparttable}

\ifCLASSINFOpdf
\else
\fi

\newtheorem{prop}{Proposition}

\begin{document}
\bstctlcite{IEEEexample:BSTcontrol}
\title{Security Region Compression via Hybrid Vertex Pruning for Enhanced Computational Efficiency}

\author{
Sicheng Gong,~\IEEEmembership{Member,~IEEE}, Zeyu Liu,~\IEEEmembership{Member,~IEEE}, \\
Danyang Xu,~\IEEEmembership{Student Member,~IEEE},
Hongjie Jia,~\IEEEmembership{Senior Member,~IEEE}
\thanks{
This work has been submitted to the IEEE for possible publication. Copyright may be transferred without notice, after which this version may no longer be accessible.

This work was supported by NEON (New Energy and mobility Outlook for the Netherlands, with project number 17628), a cross-over project financed by NWO (Dutch Research Council). S. Gong is with Department of Electrical Engineering, Eindhoven University of Technology, 5612AZ Eindhoven, Netherlands and  Department of Wind and Energy Systems, Technical University of Denmark, 2800 Lyngby, Denmark. (e-mail: s.gong@tue.nl, gongsi@dtu.dk). Z. Liu, D. Xu and H. Jia are with School of Electrical and Information Engineering, Tianjin University, 300072 Tianjin, China (e-mail: tjulzy@tju.edu.cn, danyangxu@tju.edu.cn, hjjia@tju.edu.cn).

}%
}

\markboth{}{}

\maketitle

\begin{abstract}
The security region (SR) specifies a joint feasible region in the operational space involving various installations within the power system. It enables dynamic allocation of capacity among these energy units, thereby unlocking additional grid potential for power delivery and mitigating congestion. By constructing a convex hull from a set of vertices to form a polytope, the SR can be effectively represented with an acceptable accuracy loss. However, the SR computation inevitably generates many low-quality vertices, which contribute little to the overall hypervolume and can be discarded. By introducing SR compression, the SR data size can be reduced, which in turn facilitates the memory-efficient computation of SR and time-efficient solution of SR-related optimization problems. Accordingly, this paper proposes a vertex pruning approach that mixes local and backward search policies to achieve effective SR compression. In particular, the backward search strategy allows previously pruned vertices to be recycled when they are re-identified as critical. Subsequently, we incorporate this approach to enhance computational efficiency in the aforementioned applications. Numerical validations are conducted on two realistic medium-voltage Dutch grid cases. The results demonstrate at most 1.65\% hypervolume loss with over 96.0\% SR data size reduction after applying the proposed compression approach. Empowered by this compression, the 8-dimensional (8D) SR becomes computationally tractable. Moreover, in relevant grid-aware computation challenges, the solving time can be reduced by at least 52.9\% with at most an 1.36\% optimality loss. All above validate both the efficacy and values of the proposed SR compression approach. %
\end{abstract}


\IEEEpeerreviewmaketitle

\section{Introduction}
\label{sec:intro}
The energy and mobility transition calls for the large-scale deployment of renewable generations and charging stations\,\cite{IEA2025,perera2023challenges,dunnett2022predicted}. Meanwhile, the capacity shortage within the distribution networks has become a bottleneck of establishing this sustainable energy system, reflected by the frequent occurrence of grid congestion~\cite{IEAcomment,karmaker2024electric,li2024impact}. In this context, besides increasing investments in grid expansion, the distribution system operator (DSO) seeks the coordinated operation of these energy units based on the existing grid assets. However, such flexibility is not infinite. Instead of characterizing this flexibility through algebraic power flow constraints, as given in Fig.\,\ref{fig:SR}, it can be represented geometrically by the SR, which defines a joint feasible region in the multidimensional operational space determined by the various energy units\,\cite{wu1982steady}. Each dimension maps the active/reactive injected power over a certain point of connection (POC).

Alias ``hosting capacity region”, ``flexibility region”, ``feasible operation region”, or “dynamic operation envelope”, the SR allows cooperative operation across the involved POCs within one network. Accordingly, in a smart grid environment, it becomes possible for certain POCs to temporarily exchange additional power when the others are not operating in heavy-load mode.\,\cite{gong2022extended,qin2016quantification,blackhall2020calculation,jiang2023feasible}. As illustrated in Fig.\,\ref{fig:SR}, if the corresponding operational point is located within the SR, the power flow validation can also be omitted. Without grid data security concerns, the DSO may provide an SR to relevant grid customers, and encourage their SR-associated coordinated operation. The flexibility created by this coordination exploits the grid potential for hosting more renewable generation and charging stations, while eliminating the need for additional grid asset investment.
\begin{figure}[]
     \centering
         \centering
         \includegraphics[width=\linewidth]{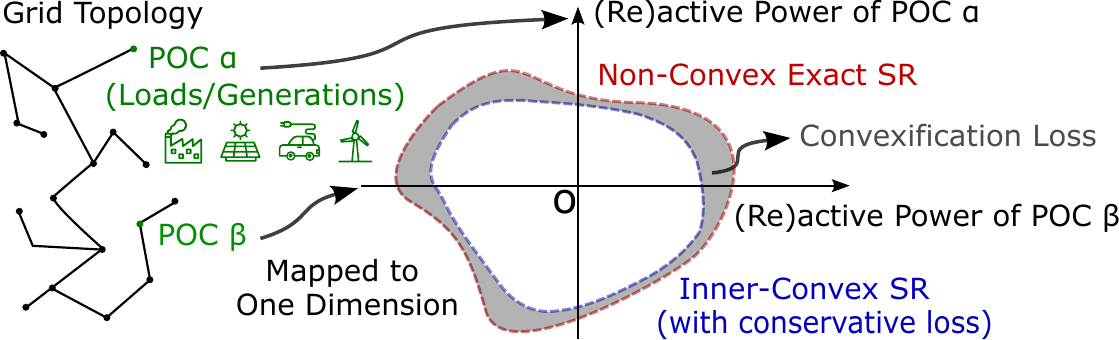}
        \caption{The concepts of SR and its variants}
        \label{fig:SR}
\end{figure}

Despite the significant industrial values that the SR may bring, the SR computation task itself is tough. There still lacks a unified and efficient scheme to ultimately solve this challenge. The SR remains computationally complicated to assess and depict, considering its non-convexity and dimensionality. To address this non-convexity, rooted in intersecting different convex variants of SR, there is an inner convexification scheme to ensure that the final result must be a subset of the original SR\,\cite{gong2023Multi}. Exemplified by the 2-dimensional (2D) instance in Fig.\,\ref{fig:SR}, this inner convexity aims to meet the risk-averse preference of the grid operator. Any grid operation or planning decision based on this result ensures the grid feasibility, thus effectively compensating for the non-convexity of SR. However, in contrast to the resolvability of the non-convexity problem, the dimensionality remains a bottleneck for computing the SR.

In order to further explain the difficulties attributed to this dimensionality, Fig.\,\ref{fig:expand} is provided,  where the convexified SR is initially a continuous body. This body should be approximated by an expansive convex polytope, thus making it discrete and processable for the computer. In a multidimensional case, this SR needs to be finally represented by a well-explored polytope, which iteratively breaks existing facets and includes more vertices for approximation accuracy control.
However, considering the curse of dimensionality, this polytope will typically contain a huge number of vertices, imposing a significant burden when storing or depicting the SR, let alone using this SR in grid applications\,\cite{de2008computational}. Therefore, the DSO anticipates a compact result. To realize this target, the SR compression should be implemented as a significant conclusive step of the SR computation. In practice, regarding a polytope-represented SR, this compression task seeks to reduce the vertex number to a fixed value, while achieving an optimal accuracy loss. It can be mathematically described as
\begin{figure}[]
     \centering
         \centering
         \includegraphics[width=0.95\linewidth]{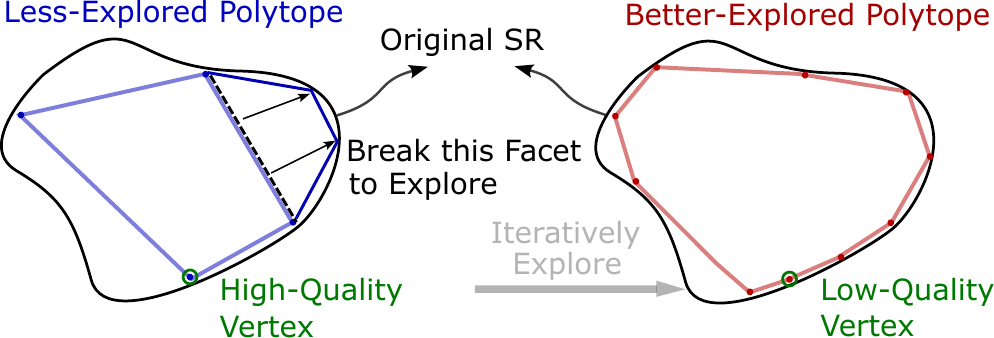}
        \caption{The expansive polytope for SR approximation}
        \label{fig:expand}
\end{figure}

\textit{Let \textbf{V} denote a vertex set, whose convex hull represents the original SR. Given the value of $n$, the problem is to derive a new vertex set $\textbf{S}$, such that
\begin{equation}
\label{eq:prime}\textbf{S}=\underset{|\textbf{S}|=n,\,\textbf{S}\subset \textbf{V}}{\textnormal{argmax}} H(c(\textbf{S}))
\end{equation}
where $|\cdot|$ measures the set cardinality. $c(\cdot)$ defines the convex hull of the given set. $H(\cdot)$ measures the hypervolume of the given convex body.}

However, if $n$ is inappropriately set to a small number, a significant loss of SR accuracy is inevitable, which may exceed the acceptance range of the DSO. Therefore, by alternatively setting a bound on this accuracy loss, we may focus on the dual problem of (\ref{eq:prime}), which is described as

\textit{Let \textbf{V} denote a vertex set. Given the value of $m$, the problem is to derive a new vertex set $\textbf{S}$, such that
\begin{equation}
    \label{eq:dual}\textbf{S}=\underset{H(c(\textbf{S}))\geq m,\,\textbf{S}\subset \textbf{V}}{\textnormal{argmin}} |\textbf{S}|
\end{equation}
where $m$ is the hypervolume threshold of the compressed SR.}

Either with the problem setting being (\ref{eq:prime}) or (\ref{eq:dual}), it can be seen as a polytope simplification challenge, which is NP-hard and has been investigated intensively in recent decades\,\cite{cohen1999concepts}. From a geometric perspective, SR compression shares essential similarities with polytope simplification, including the objectives of vertex number reduction and hypervolume loss control. The polygon reduction approaches using sequential edge collapses has been developed to simplify the polytope, though they rely on the precondition of polygon triangulation\,\cite{melax1998simple,veron1998shape}.  Based on the similar precondition, a mesh partition and recombination method has been developed to compress a large polygon\,\cite{ho2001compressing}. However, both compression methods aims to compress the non-convex polygon, and the convex-polytope-represented SR is already the most concise result for the above methods.
Another technical path of polytope simplification is geometric representation by simplified regularized shapes, for instance ellipsoid, zonotope, and orthotope\,\cite{bemporad2004inner,kousik2022ellipsotopes}. This definitively reduces the representative complexity and guarantees an inner approximation result. However, the hypervolume loss is suboptimal, which may significantly harm the SR values of grid potential exploitation. The DSO anticipates a convex polytope simplification method with optimized hypervolume loss, thereby realizing the SR compression. To address this challenge, this paper proposes a vertex pruning approach to conduct this compression task. The main contributions of this paper are fourfold:
\subsubsection{Theoretical analysis} Beyond clarifying the importance of SR compression, we prove the mathematical equivalence between optimal compression and vertex pruning. This equivalence supports a compact and tangible reformulation of the primal SR compression problem in theory.
\subsubsection{Compression algorithm development} To balance computational optimality and complexity in SR compression, we develop a hybrid vertex pruning approach by tailoring a novel backward search policy and integrating the variable neighborhood search (VNS) method. This compensates for the inherent short-sighted deficiency of the local greedy policy.

\subsubsection{Computation efficiency enhancement} We incorporate the SR compression approach into the relevant computational tasks. This allows SR to be computed in a memory-efficient manner to improve its dimensionality, and also supports the time-efficient solution of SR-related optimization problems.

\subsubsection{Validation on realistic scenarios} Based on two MV grid cases using realistic charging station settings, relevant numerical studies are conducted to confirm the efficacy of the proposed approach, together with the computational efficiency gains brought by the SR compression.

The remainder of this paper is organized as follows. Section~\ref{sec:motivation} explains the motivations of SR compression, including its feasibility and benefits. Afterwards, a hybrid vertex pruning approach is proposed in Section~\ref{sec:method}. The method validations by experimental tests are conducted
in Section~\ref{sec:validation} and a conclusion is provided in Section~\ref{sec:conclusion}.

\section{Security region compression}
\label{sec:motivation}

\subsection{The observant inevitability of low-quality vertices}
\label{subsec:inevitable}
As shown in Fig.\,\ref{fig:compress}, the feasibility of SR compression is revealed by the fact that a part of the vertices can be removed from the polytope with minimal loss of hypervolume. In fact, as depicted by Fig.\,\ref{fig:expand}, throughout the region computation process, low-quality vertices are inherently produced since the SR attempts various exploratory directions. With a low-quality vertex included, the SR hypervolume addition is small. Meanwhile, when a high\mbox{-}quality solution is identified and added, the resulting increase in hypervolume can be substantial. With the iteration of multiple low-quality vertices and one high-quality vertex, the entire SR expansion process eventually leads to a stair-wise hypervolume curve as depicted in Fig.\,\ref{fig:quality}.

By observing a series of SR computation experiments in \cite{gong2023Multi,gong2024rapid}, we empirically conclude that the exploration direction must be iteratively adjusted until a significant SR hypervolume gain is achieved. This process typically yields an initial sequence of low-quality vertices. Therefore, to ensure a trustworthy SR assessment result, an excessive number of low-quality vertices are inevitable to be included. This phenomenon highlights the feasibility and necessity of SR compression.
Notably, as exemplified in Fig.\,\ref{fig:expand}, the role of high-quality vertex is unstable, as it may degrade into a low-quality one with the polytope expanding.
\begin{figure}[]
     \centering
     \begin{subfigure}[b]{0.24\textwidth}
         \centering
         \includegraphics[width=\textwidth]{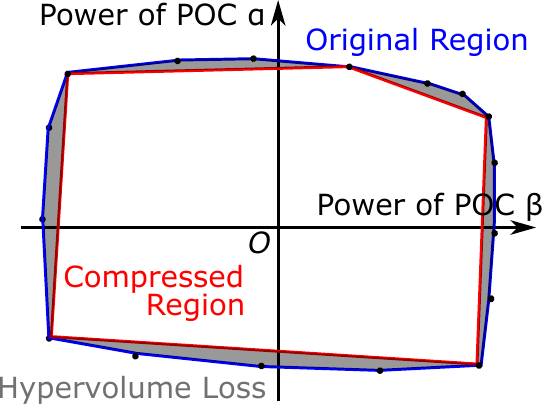}
         \caption{Compression concept}
         \label{fig:compress}
     \end{subfigure}
     \begin{subfigure}[b]{0.24\textwidth}
         \centering
         \includegraphics[width=\textwidth]{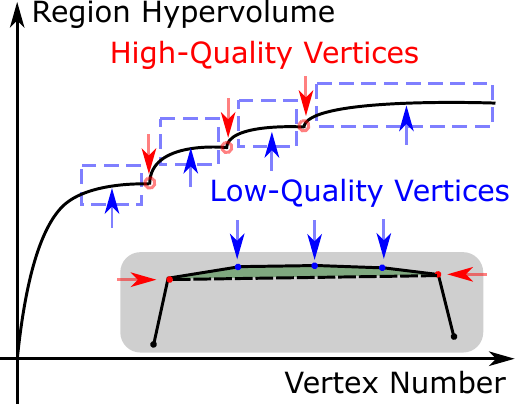}
         \caption{Vertex quality concept}
         \label{fig:quality}
     \end{subfigure}
        \caption{Concepts in SR compression}
        \label{fig:concept}
\end{figure}

\subsection{The benefits of SR compression}
\subsubsection{Memory space saving }
In theory, a vertex set is enough to represent the SR. Meanwhile, in modern computer-aided grid operation practices, the SR objective should store more pre-calculation results, which will be frequently recalled for extended applicability, for instance continued expansion or vertex pruning.  In reference to the SciPy Library, Table\,\ref{tab:data} lists the concrete data structure of SR, where the sizes of most sub-objects depend on the vertex number $N$ and the dimension number $d$\,\cite{2020SciPy-NMeth}. Apart from the hypervolume value that is stored as a single floating-point number, the memory space required to store the SR objective will decrease linearly with the reduction of $N$.

\begin{table}[ht]
    \centering
    \caption{Data structure of an SR object}
    \begin{tabular}{c c c}
        \toprule
        Sub-Object  & Type & Size  \\
        \midrule
         Vertex sequence & Array & $[N,d]$ \\
         Simplex sequence &  Array & $[(d-1)(N-2),d]$ \\
         Facet-based hyperplane sequence &  Array & $[(d-1)(N-2),d+1]$ \\
         Hypervolume &  - & $1$ \\
         \bottomrule
    \end{tabular}
    \label{tab:data}
\end{table}

\subsubsection{Improved dimensionality of SR computation}
The above benefit is based on the assumption of initially abundant memory space for the uncompressed SR. In the current situation of memory chip shortage, facing a computer with limited memory space, the DSO needs to synchronize region expansion and compression to avoid potential memory overflow\,\cite{xiong2025semiconductor}. On the edge of exhausting memory in the process of SR expansion, we can conduct the compression task and then restore the expansion process. This prevents the memory space from being depleted during the entire SR computation process. In practice, limited memory space has become a bottleneck for SR dimensionality. With this synchronous expansion-compression framework, using the same scientific computation asset, a higher-dimension SR calculation task becomes executable. To this end, based on the SR computation scheme developed in \cite{gong2023Multi}, Algorithm\,\ref{Ag:dimension} is developed to realize the SR computation with the controllable memory space requirement, eventually improving the computational dimensionality. The SR compression method in Line 8 at Algorithm\,\ref{Ag:dimension} is adopted as Algorithm\,\ref{Ag:hybrid}, which will be later investigated in Section\,\ref{subsubsec:hybrid}.

\begin{algorithm}[t]
\KwData{Power flow constraint $\textbf{PF}$, vertex number upper bound $M$, vertex number lower bound $\hat{M}$, dimension number $D$.}
Initialize $D+1$ marginal points based on $\textbf{PF}$, and group them into a set $\textbf{V}$ \;
Initialize the polytope $c(\textbf{V})$ for SR representation\;
\Repeat{Stopping criterion reaches}{
\Repeat{$|\textbf{V}|\geq M$}{
Determine SR expansion direction $d$ using  $c(\textbf{V})$. Following $d$, calculate the marginal point $v$ through Bisection method \;
Add $v$ to $\textbf{V}$ \;
}
Run the SR compression to ensure $|\textbf{V}|=\hat{M}$ \;
}
 \caption{Memory-efficient SR computation}
\label{Ag:dimension}
\end{algorithm}

\subsubsection{Complexity reduction of grid-aware optimization problems}
\label{subsubsec:gridaware}
The SR is able to play an alternative role of power flow constraints in grid-aware optimization problems. Accordingly, a compressed SR can contribute to a more compact problem formulation, thereby enhancing the solution speed. To exemplarily illustrate this benefit, two optimization instances are provided, which will be later adopted for further numerical validations in Section\,\ref{sec:validation}.

\paragraph{Optimal energy storage sizing for grid feasibility retrieval} In situations of grid congestion, operational limits of the grid are exceeded because the existing grid customers draw or inject too much power. To retrieve power flow feasibility, energy storage appliances are planned to be installed, thereby locally smoothing the power input and output of the associated energy units\,\cite{wang2014power}. From a geometric standpoint, as exemplarily depicted by the 2D case in Fig.\,\ref{fig:sizing}, the initial operating point lies outside the SR, and the objective of the energy storage equipment is to shift this operating point back into the SR. Consequently, DSO seeks to determine the optimal size of these storage systems, such that the summative axial distance between the original and shifted operation point is minimized\,\cite{gong2023Multi}.

The corresponding mathematical programming model can be formulated as  (\ref{eq:region_op}), where the SR judges whether the new operational point shifted by the energy storage equipment becomes acceptable for the grid. Following the target of optimal sizing, the objective function is defined as the absolute aggregate of the regulating power supplied by the storage equipment. The complexity of solving (\ref{eq:region_op}) highly depends on the number of vertices in SR, as a smaller-size SR will lead to a lower number of linear constraints in (\ref{eq:region_op}). For the given uncompressed SR, its compression can significantly help reduce the solving complexity of (\ref{eq:region_op}).
\begin{equation}
\label{eq:region_op}
\begin{aligned}
    \min~&{\sum_{i\in \textbf{G}_a} |\hat{P}_i -P_i|+ \sum_{j\in \textbf{G}_r} |\hat{Q}_j -Q_j|}  \\
    s.t.~& (\textbf{P}, \textbf{Q}) \in \textbf{R}\\
    & \textbf{P}=\{P_i\,|\,i\in \textbf{G}_a \}\,,\,\textbf{Q}=\{Q_j\,|\,j\in \textbf{G}_r \}
\end{aligned}
\end{equation}
where $\textbf{G}_a$ denotes the set of POCs whose active power is taken as one dimension of the SR. $\hat{P}_i$ is the aggregated active  power over a certain POC before the installation of storage equipment. ${P}_i$ is an auxiliary decision variable, reflecting the active power over this POC after regulation. Its absolute difference from $\hat{P}_i$ denotes the regulating power. $\textbf{G}_r$, $\hat{Q}_j$, and ${Q}_j$ are defined likewise in the reactive power domain. $(\textbf{P}, \textbf{Q})$ denotes a multidimensional point that reflects the active and reactive power over relevant POCs. $\textbf{R}$ denotes its corresponding SR.

\begin{figure}[]
     \centering
     \begin{subfigure}[b]{0.24\textwidth}
         \centering
         \includegraphics[width=\textwidth]{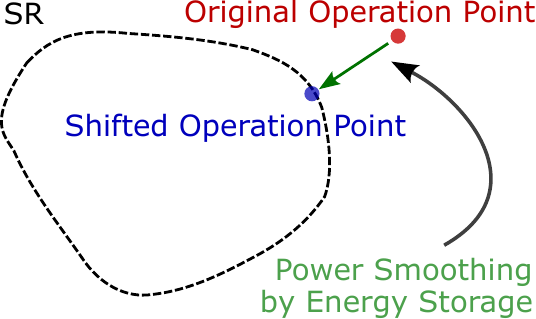}
         \caption{Energy storage sizing}
         \label{fig:sizing}
     \end{subfigure}
     \begin{subfigure}[b]{0.24\textwidth}
         \centering
         \includegraphics[width=\textwidth]{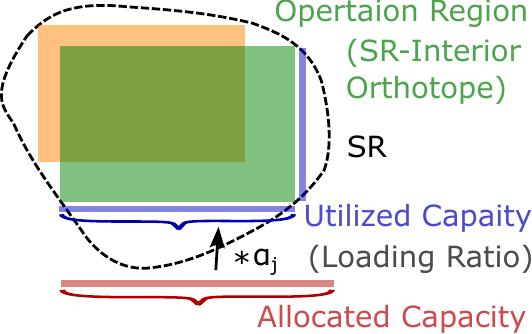}
         \caption{Capacity allocation}
         \label{fig:allocate}
     \end{subfigure}
        \caption{SR-related optimization problems}
        \label{fig:optim}
\end{figure}

\paragraph{Optimal capacity allocation under capacity shortage}
The second case is to maximize the total grid integration capacity assigned to various energy units, thus exploiting grid utility, particularly when the grid is on the edge of congestion. From a geometric perspective, as depicted by Fig.\,\ref{fig:allocate}, this corresponds to maximizing the sum of the edge lengths of an orthotope, whose shrunk version should be an SR-interior orthotope-shaped operation region. The loading ratio can be set to a value around 100\%, because it is rare that all participating energy units operate at full power simultaneously. Meanwhile, certain units may temporarily operate above their capacity limits, although they may later incur penalties from the DSO.

For simplicity of presentation, only active power is considered in this paper. Consequently, the optimization problem can be formulated as (\ref{eq:region_gp}), and this formulation can be readily extended to incorporate reactive power.
\begin{subequations}
\label{eq:region_gp}
\begin{align}
    \max~&{\sum_{i\in \textbf{G}_a} \Delta P_i} \\
    \label{eq:region_gp_2}
    \begin{split}
    s.t.~& (P_1+B_1 \alpha_1 \Delta P_1,\cdots,P_{m}+B_{m} \alpha_m \Delta P_{m})\in \textbf{R}
    \\
    &  (B_1,\cdots,B_{m})\in \{0\,,\, 1\}^{m} \,,\, m=|\textbf{G}_a|
    \end{split}
    \\
    & \Delta P_i\geq 0 \,,\, i\in \textbf{G}_a \,,\, \alpha_i\geq \Hat{\alpha}_i
\end{align}
\end{subequations}
where  $\Delta P_i$ denotes the active power capacity allocated to a certain POC for the incoming energy unit. $|\textbf{G}_a|$ and $\textbf{R}$ share the same definitions with those in (\ref{eq:region_op}). $|\cdot|$ measures the set cardinality and $B_i$ is a binary variable. $(P_1,\cdots,P_{m})$ represents the lexicographical minimum vertex of the targeted orthotope, which contains $2^{m}$ vertices. $\alpha_i$ measures the occupied ratio of capacity, and $\Hat{\alpha}_i$ sets a boundary value to ensure the grid redundancy, to avoid the risk-aggressive capacity allocation result. Therefore, (\ref{eq:region_gp_2}) is composed of $2^{|\textbf{G}_a|}$ constraints, each of which represents that a vertex of the targeted orthotope should be located in the SR. In analogy to (\ref{eq:region_op}), the  complexity of solving (\ref{eq:region_op}) also depends on the vertex number of $\textbf{R}$, thereby promoting SR compression.

\section{Vertex pruning approach design}
\label{sec:method}
Although the feasibility and benefits of SR compression have been clarified above, compared to vertex pruning, it is more of a target than an approach. This section starts with discussing the metathetical relationship between SR compression and vertex pruning. After clarifying their equivalent relationship, concrete vertex pruning approaches are investigated for improved SR compression performance.
\begin{figure}[]
     \centering
     \begin{subfigure}[b]{0.24\textwidth}
         \centering
         \includegraphics[width=\textwidth]{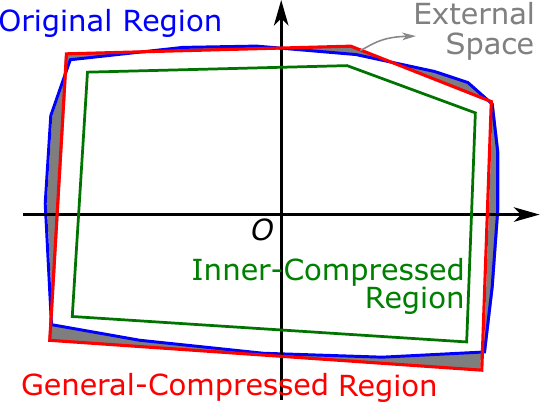}
         \caption{Compression defintions}
         \label{fig:Convention}
     \end{subfigure}
     \begin{subfigure}[b]{0.24\textwidth}
         \centering
         \includegraphics[width=\textwidth]{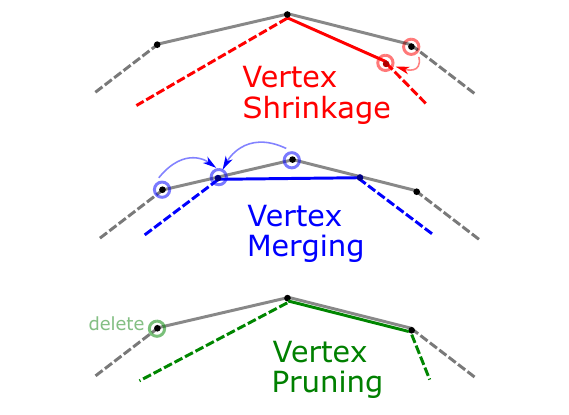}
         \caption{Inner-compression types}
         \label{fig:Inner}
     \end{subfigure}
        \caption{Further illustrations on compression types}
        \label{fig:CompressionType}
\end{figure}

\subsection{Equivalence of SR compression and vertex pruning}
The general definition of ``compression" aims to achieve a high-similarity result with a reduced data volume, indicating a different problem setting from (\ref{eq:prime}) or (\ref{eq:dual}). As shown in Fig.\,\ref{fig:Convention}, the general-compressed region with less vertices may contain external space compared to the original one. Meanwhile, in practice, to reserve the risk-averse property that is preferred by the grid operator, the compressed SR should be a subset of the original polytope, indicating the necessity of an inner-compressed SR. To prevent excessive discussions, the primal inner-compression problem will mainly be investigated in this paper. Inheriting the notations of (\ref{eq:prime}), this  compression task can be mathematically described as

\textit{Let \textbf{V} denote a vertex set, whose convex hull represents the original SR. Given the value of $n$, the problem is to derive a new vertex set $\textbf{S}$, such that
\begin{equation}
\label{eq:prime_compression}
\textbf{S}=\underset{|\textbf{S}|=n\,,\,c(\textbf{S})\subset c(\textbf{V})}{\textnormal{argmax}} H(c(\textbf{S}))
\end{equation}
}

Intuitively, as shown in Fig.\,\ref{fig:Inner}, (\ref{eq:prime_compression}) may lead to three kinds of compression results, including vertex shrinkage, merging, and pruning. Accordingly, (\ref{eq:prime_compression}) owns a larger solution space while sharing the same objective with (\ref{eq:prime}). However, by inheriting the notations in (\ref{eq:prime}), we have Proposition\,\ref{prop:eq}. It argues that (\ref{eq:prime}) is strictly equivalent to (\ref{eq:prime_compression}), and (\ref{eq:prime}) is a compact reformulation of (\ref{eq:prime_compression}). In other words, regarding a polytope-represented SR, the new vertex set that formulates the optimally compressed SR, must be a subset of the original vertex set.
By analogy, the dual format of (\ref{eq:prime_compression}) can also be proved equivalent to (\ref{eq:dual}) by contradiction. To this end, we claim the equivalence of SR compression and vertex pruning, thus making the concept of SR compression tangible.

\begin{prop}
\label{prop:eq}
$\exists \textbf{V}_{o}\in \big\{\,\textbf{S}\,|\,\textbf{S}=\textnormal{argmax}\big[H(c(\textbf{S}))\big],|\textbf{S}|=n\,,\, c(\textbf{S})\subset c(\textbf{V})\,\big\}$\,,\,$\textbf{V}_{o}\subset \textbf{V}$.
\end{prop}

\begin{proof}
We prove it by contradiction. If $\textbf{V}_{o}$ must contain at least one point $i$ that $i\notin \textbf{V}$, $i$ might have three different situations, including being outside the polytope, strictly inside the polytope (excluding the margin), and on the margin:

(1) $i\notin c(\textbf{V})$: Since $i\in \textbf{V}_{o}$, it contradicts $c(\textbf{V}_{o}) \subset c(\textbf{V})$\,.

(2) $i\in c(\textbf{V})$\,;\,$\forall F\in \textbf{F}, i\notin F$: $\exists j \in \textbf{V}_{o-}$, extends the vector $(j,i)$ to configure a intersection point $k$ with $F$. Then $H(c(\textbf{V}_{o-}\cup k))>H(c(\textbf{V}_{o}))$, violating the optimality of $\textbf{V}_{o}$\,.

where $\textbf{F}:=\{ c(\textbf{V}_{-1})\,|\, \textbf{V}_{-1}\subset \textbf{V}, |\textbf{V}_{-1}| = d-1\}$, d is the polytope dimension number,
$\textbf{V}_{o-}:= \{v\,|\,v\in\textbf{V}_{o},v\neq i \}$\,.

(3) $i\in c(\textbf{V})$\,;\,$\exists F\in \textbf{F}, i\in F$: $\exists \{w_1, w_2,\cdots, w_{d-1}\}$, such that $i=\sum\limits_{m=1}^{d-1}w_m\cdot j_m$, where $\sum\limits_{m=1}^{d-1}w_m=1, w_m>0$. Then
\begin{subequations}
    \begin{align}
        H(c&(\textbf{V}_{o})) = H(c(\textbf{V}_{o-}))+ H(c(\textbf{K}\cup i)) \\
        &= H(c(\textbf{V}_{o-}))+ \frac{H(F_k)}{d}\cdot D(F_k,i) \\
        \label{subeq:superposition}
        &=H(c(\textbf{V}_{o-}))+ \frac{H(F_k)}{d}\sum\limits_{m=1}^{d-1}w_m\cdot D(F_k,j_m) \\
        &=\sum\limits_{m=1}^{d-1}w_m\big[ H(c(\textbf{V}_{o-}))+ \frac{H(F_k)}{d}\cdot D(F_k,j_m)\big]\\
        &= \sum\limits_{m=1}^{d-1}w_m\cdot H(c(\textbf{V}_{o-}\cup j_m)) \\
        \label{subeq:leq}
        &\leq \underset{m\in\textbf{Z}^{+},m<d}{\max}\{H(c(\textbf{V}_{o-}\cup j_m))\}
    \end{align}
\end{subequations}
where $\{j_1,j_2,\cdots,j_{d-1}\}$ defines the vertex set of $F_i$ considering $i\in F_i$ and $F_i\subset \textbf{F}$. In analogy to $\textbf{F}$, $\textbf{F}_o$ and $\textbf{F}_{o-}$ define the facet sets of $c(\textbf{V}_{o})$ and $c(\textbf{V}_{o-})$. $F_k:=\textbf{F}_o \setminus \textbf{F}_{o-}$, which is the replaced facet of $c(\textbf{V}_{o-})$ after adding $i$ to $\textbf{V}_{o-}$ for polytope expansion.  $\textbf{K}$ is the vertex set of $F_k$. $D(\cdot,\cdot)$ calculates the distance between the facet-defined hyperplane and the point, which satisfies the superposition principle for points that are located in the same hyperplane-defined halfspace. This explains the validity of (\ref{subeq:superposition}). $\textbf{Z}^{+}$ is the positive integer set.

The inequality relationship in (\ref{subeq:leq}) contradicts the optimality of $\textbf{V}_{o}$. The equality relationship in (\ref{subeq:leq}) becomes valid only when $D(F_k,j_m)$ remains the same for all $j_m$, indicating that $F_k$ is in parallel with $F_i$. Hence, $\textbf{V}_{o-}\cup j_m$ can be taken as another optimal solution. Inductively, we can find an optimal solution where all points belong to $\textbf{V}$, which contradicts the basic pre-assumption.

\end{proof}

\subsection{Step-by-step and one-off patterns}
The design of concrete pruning policy is based on the selection of pruning pattern. There are two patterns including the step-by-step and one-off ones. The front pattern may determine one vertex to prune in each iteration, while the latter one prunes a group of vertices at once.
If the problem scale is processable for the grid operator, an one-off pruning pattern becomes recommendable due to its horizontal optimality. Meanwhile, given a large-scale pruning problem, we recommend a step-by-step pattern for two reasons:
\paragraph{Controllable computation burden} The one-off pattern asks to solve (\ref{eq:prime}) directly after determining the value of $n$, whose computational hardness is emphasized by its NP-hard essence. Since $H(\cdot)$ is hard to write in an analytical format, the traversal process is inevitable to acquire an optimal solution to (\ref{eq:prime}). With $M$ denoting the value of $|\textbf{V}|$, the summative computational complexity is up to $\frac{(M-1)!M^2logM}{n!(M-n)!}$, considering the single-round computational complexity of $H(\cdot)$ is $MlogM$. Instead, the step-by-step pattern can be controllable at each iteration, whose stepwise computation burden is no larger than $M^2logM$ if adopting the greedy strategy. To this end, the summative overestimated computational complexity is reduced to $\frac{M(M+1)(M+2)logM}{6}$, which is much lower than that of the one-off pattern.

\paragraph{Transparent and resilient process} The step-by-step pattern contributes to a transparent pruning process for ease of monitoring. With limited computational resources and time, the pruning process can be flexibly suspended or restarted with progressive pruning decisions. This eliminates the risk of missing or low-quality interim results under interruptions if adopting the one-off pattern.  Moreover, the step-by-step pattern renders the selection between the problem settings of (\ref{eq:prime}) and (\ref{eq:dual}) unnecessary, as the stopping criterion can adaptively alternate between the vertex number and the region loss. This provides additional administrative flexibility when executing the SR compression task.

\subsection{Concrete vertex pruning approaches}
Based on the step-by-step pattern, two fundamental pruning policies are designed. Furthermore, a hybrid approach that alternates between these two policies is proposed to achieve a trade-off between decision complexity and optimality.
\subsubsection{Local search policy}
By following the greedy principle, the local search policy makes a locally optimal decision at each step, which is equivalent to the solution of (\ref{eq:prime}) by setting $n$ to be $|\textbf{V}|-1$.
By defining $N$ as the current SR vertex number at a certain step, the computational complexity of this policy is as low as $O(N^2 logN)$, which is contributed by a quick sort and the N rounds of $H(\cdot)$ computation. However, the greedy nature of this local search policy harms the long-term optimality of the entire pruning process, which can be demonstrated by a counterexample provided in Appendix\,\ref{app:counter}. Meanwhile, facing the finite computational capacity of the DSO, a suboptimal solution is acceptable, as long as its solution quality remains high and contributes to a significant reduction of the computational burden.

\subsubsection{Backward search policy}
To compensate for the loss of long-term optimality, a backward search policy is devised to improve decision quality in the long-term vision. As illustrated in Fig.\,\ref{fig:back}, Vertex\,1 is pruned by the local search policy at the first step. Afterwards, using the backward search policy, regarding Vertex\,5, we seek an alternative vertex within its neighborhood, particularly focusing on the pruned vertices to replace Vertex\,5 in the end. Eventually, Vertex\,1 is selected and recycled, with Vertex\,5 being pruned. In essence, without the direct participation in the compression process, this backward policy aims to figure out a larger SR by reassessing the pruned vertices and recycling the valuable ones. Compared with the local search policy, this backward one ensures a higher-quality result in the price of increased computational effort. Quantitatively, with $\hat{M}$ denoting the neighborhood set cardinality, the computational complexity of this backward policy at each iteration is $O(\hat{M}\cdot N^2logN)$. As $M$ is set higher in this policy, we can achieve a better recycling outcome, though at the cost of higher computational expense.

Aiming for an appropriate balance between the gain of backward search and its associated computational cost, the VNS metaheuristic has been incorporated \,\cite{hansen2010variable}. The initial phase of this procedure is identical to the original search policy, and it determines an optimal vertex $v_{n}$ within the neighborhood of the candidate vertex $v_o$. In the subsequent phase, $v_n$ is substituted for $v_o$, and a new optimal vertex $v_f$ is then identified within the neighborhood of $v_n$. Within the operations research literature, the VNS approach has been demonstrated to constitute an effective heuristic framework for trading-off solution quality against computational complexity. Specifically, it asymptotically identifies an optimal candidate within the $O(M^2)$-size neighborhood set, while constraining the summative search set size to $O(2M)$. Consequently, relative to the original backward search policy described above, this VNS-enhanced version seeks to obtain a higher-quality recycled vertex while keeping its computational complexity under control.

\begin{figure}[]
     \centering
      \begin{subfigure}[b]{0.24\textwidth}
         \centering
         \includegraphics[width=\textwidth]{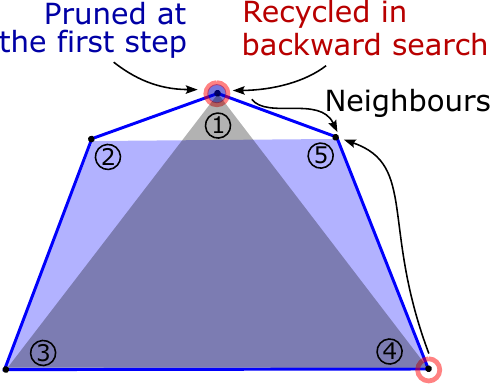}
         \caption{Backward search policy}
         \label{fig:back}
     \end{subfigure}
     \begin{subfigure}[b]{0.24\textwidth}
         \centering
         \includegraphics[width=\textwidth]{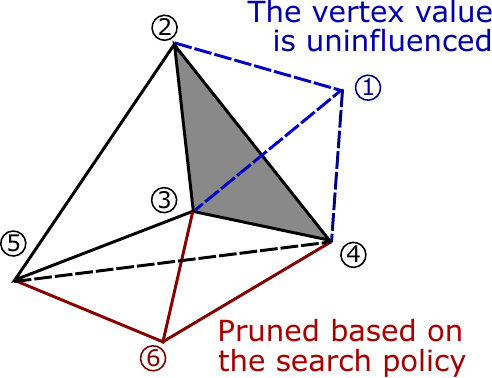}
         \caption{Search policy acceleration}
         \label{fig:facet}
     \end{subfigure}
        \caption{Vertex pruning policy and acceleration}
        \label{fig:acceleration}
\end{figure}
\begin{algorithm}[t]
\KwData{Vertex list $\textbf{V}$, neighborhood set $\textbf{N}(v)$ of vertex~$v$, the policy alternation threshold $T$.}
Initialize the $\textbf{V}$-respective hypervolume loss list $\textbf{L}$, $\textbf{L}[i]= H\big(c(\textbf{V})\big)-H\big(c(\textbf{V}\setminus \textbf{V}[i] )\big)$ \textit{\Big($H\big(c(\cdot)\big)$ returns the convex hull hypervolume of the vertex set\Big)} \;
\Repeat{stopping criterion reaches }{
\If{$l\geq T $}{
Pop $\textbf{V}[i]$ from $\textbf{V}$, $i=\textnormal{argmin}\{\textbf{L}[i]\}$ \;
$l= \textbf{L}[i]$, Update $\textbf{L}$ \;}
\Else{
Let $v_o=\textbf{V}[i]$, $i=\textnormal{argmin}\{\textbf{L}[i]\}$ \;
$\textbf{V}_{-} = \textbf{V} \setminus v_o$, $\textbf{N}_{o} =\textbf{N}(v_o)$, $v_n=v_o$ \;
\For{$v$ in $\textbf{N}_o[i]$}{
$\Delta= H\big(c(\textbf{V}_{-}\cup v )\big)-H\big(c(\textbf{V}_{-}\cup v_o )\big)$ \;
\If{$\Delta \geq 0$}
{Replace $v_n$ by $v$ \;}
}
$\textbf{V}_{-} = \textbf{V} \setminus v_n$, $\textbf{N}_{n}=\textbf{N}(v_n)$, $v_f=v_n$ \;
\For{$v$ in $\textbf{N}_{n}[i]$}{
$\Delta= H\big(c(\textbf{V}_{-}\cup v)\big)-H\big(c(\textbf{V}_{-}\cup v_n )\big)$ \;
\If{$\Delta \geq 0$}
{Replace $v_f$ by $v$ in $\textbf{V}$ \;
$\textbf{L}[i]= (\textbf{L}[i]+\Delta)$, $l= \textbf{L}[i]$ \;}
}
}
}
 \caption{Hybrid vertex pruning}
\label{Ag:hybrid}
\end{algorithm}

\subsubsection{The hybrid pruning approach}
\label{subsubsec:hybrid}
Based on two search policies above, by integrating the VNS method, a hybrid vertex pruning approach has been proposed as given in Algorithm\,\ref{Ag:hybrid}, where the processed vertex list $\textbf{V}$ represents the compressed SR. For further illustrations, in the initial compression stage, the local search policy is recommended for quick low-quality vertex pruning. Afterwards, the backward policy is integrated to formulate a hybrid approach for pruning vertices more carefully. To this end, a hybrid scheme is established to alternate between local and backward search policies, aiming for eventual high-quality pruning decisions with an acceptable computational burden. A threshold on the compression-induced hypervolume loss can be set, crossing which will force a integration of the backward search policy. Notably, the update of $\textbf{L}$ in Algorithm\,\ref{Ag:hybrid} can be accelerated. As illustrated in Fig.\,\ref{fig:facet}, if Vertex\,6 is pruned, the value of Vertex\,1, indicating the potential hypervolume loss after pruning this vertex, remains unchanged. Therefore, the corresponding element in $\textbf{L}$ is unnecessary to recalculate, thus accelerating the entire pruning process.

\section{Experimental results}
\label{sec:validation}

Two Dutch realistic 10.5 kV grid cases, namely a 9-bus system and a 296-bus system, are used to numerically verify the effectiveness of the proposed methods. The relevant grid parameters have been listed as open-source in the appendix of \cite{gong2025thesis}, which will no longer be repeated in this paper. Focusing on the active power, by referring to the SR computation results in \cite{gong2023Multi} and \cite{gong2024rapid}, the SRs of POC\,2,4,7 and POC\,3,5,7 in the 9-bus system are taken as two testing instances, which separately own 477 and 468 vertices. Regarding the 296-bus case, the SR of POC\,2,25,27 is taken as the additional testing instances with 504. These three polytope-represented SRs have at most 2.14\% hypervolume loss compared to the exact ones, and they are anticipated to be further compressed for data size reduction. All codes ran on a computer with an Apple M4 processor (10-core) and 16\,GB of RAM, operating under MacOS\,26 version.

\subsection{The performance analysis in SR compression}
The SR compression is implemented through running Algorithm\,\ref{Ag:hybrid}. The neighborhood set cardinality for VNS realization is set to 5, which corresponds to $|\textbf{N}(v)|$ in Algorithm\,\ref{Ag:hybrid}.

\subsubsection{Results in the 9-bus testing system}
In the instance of POC\,2,4,7, the hypervolume loss threshold is set at 1\%, which may trigger the substitution of the local search policy with the backward one. The result demonstrates that the original SR has already been reduced to a polytope of 15 vertices before this substitution. Accordingly, to keep coherent compression performance analysis, regarding the other instance of POC\,3,5,7, we set the threshold at 0.5\%, such that the compressed SR also has 15 vertices before switching the policy.

\paragraph{The first phase of hybrid vertex pruning}
During the initial compression stage, by only utilizing the local search policy, both SRs continue to be compressed, and the relative values of their cumulative hypervolume loss have been plotted in Fig.\,\ref{fig:initial}. Though both SR data sizes have been reduced by at least 96.7\%, they suffer at most 1.25\% hypervolume loss. To further illustrate the representative conciseness attributed to the compression, the original and compressed SRs in both testing instances have been depicted in Fig.\,\ref{fig:compareSR}.

\begin{figure}[]
     \centering
         \centering
         \includegraphics[width=0.95\linewidth]{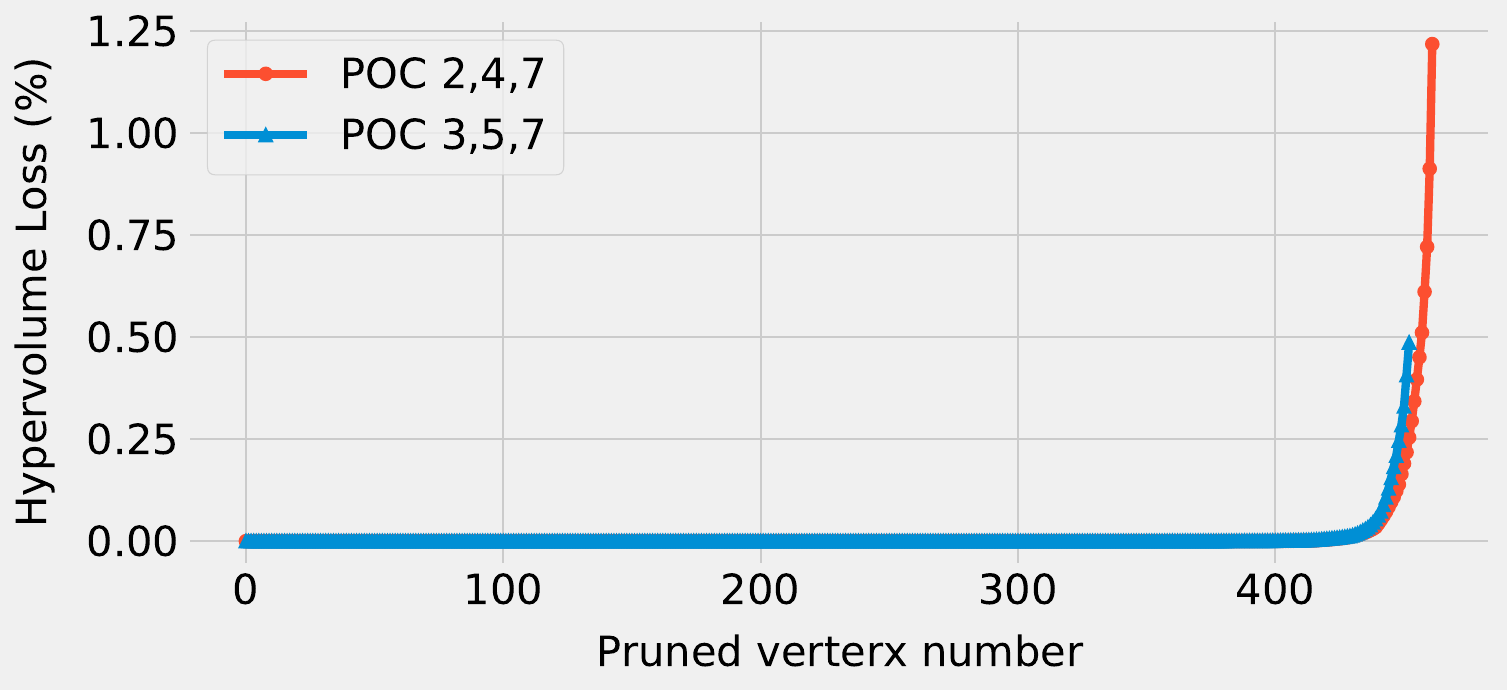}
        \caption{The compression performance of local search policy}
        \label{fig:initial}
\end{figure}

In fact, when both SRs are pruned to 50-vertex polytopes, the relative total hypervolume loss remains lower than 0.08\textperthousand. This supports our claim that the ratio of low-quality vertices is higher than 89\% in both instances, which matches the empirical conclusion regarding the inevitable presence of low-quality vertices as stated in Section\,\ref{subsec:inevitable}.

\begin{figure}[]
     \centering
     \begin{subfigure}[b]{0.24\textwidth}
         \centering
         \includegraphics[width=\textwidth]{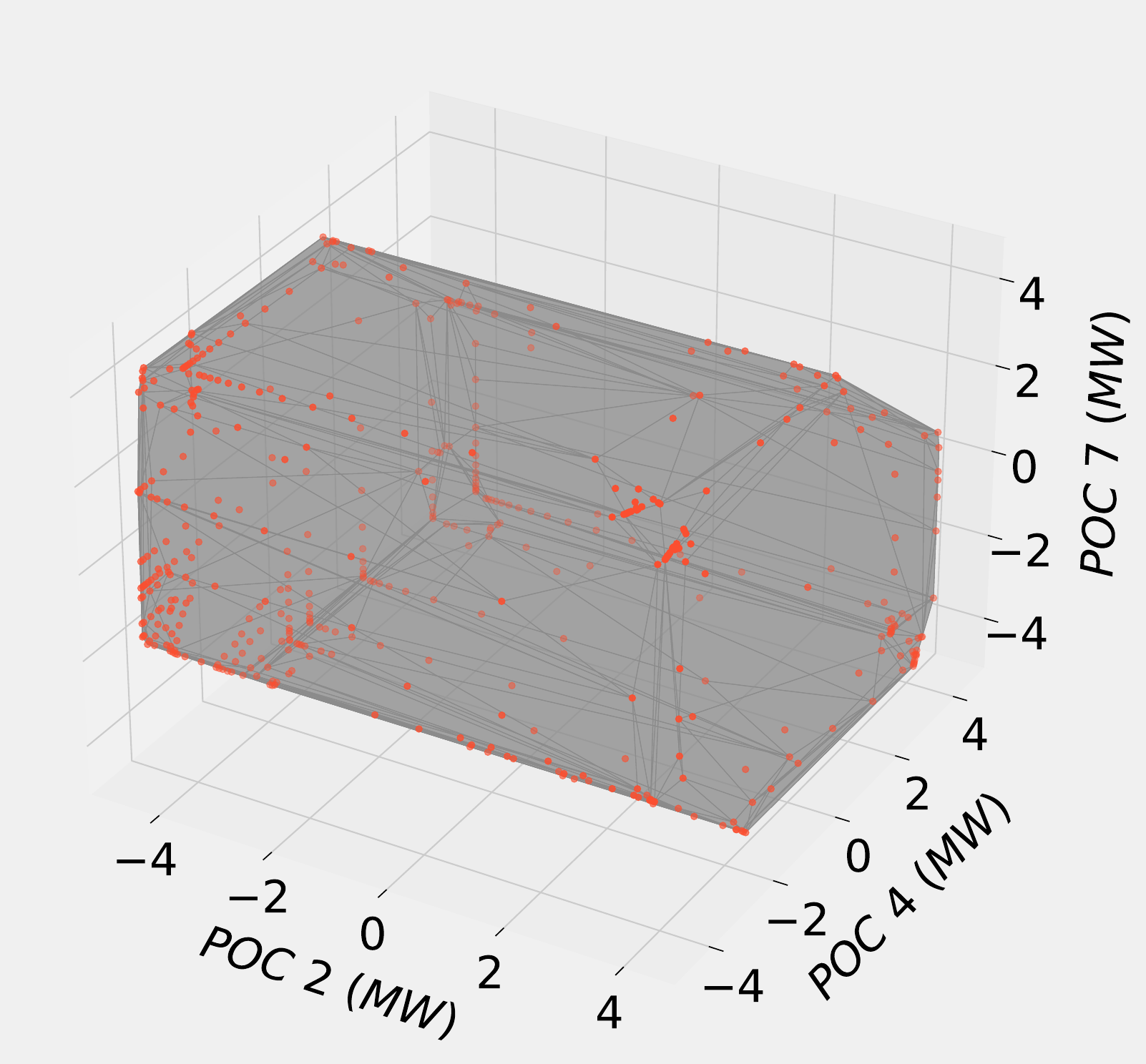}
         \caption{477-vertex SR of POC\,2,4,7}
         \label{fig:477SR247}
     \end{subfigure}
     \begin{subfigure}[b]{0.24\textwidth}
         \centering
         \includegraphics[width=\textwidth]{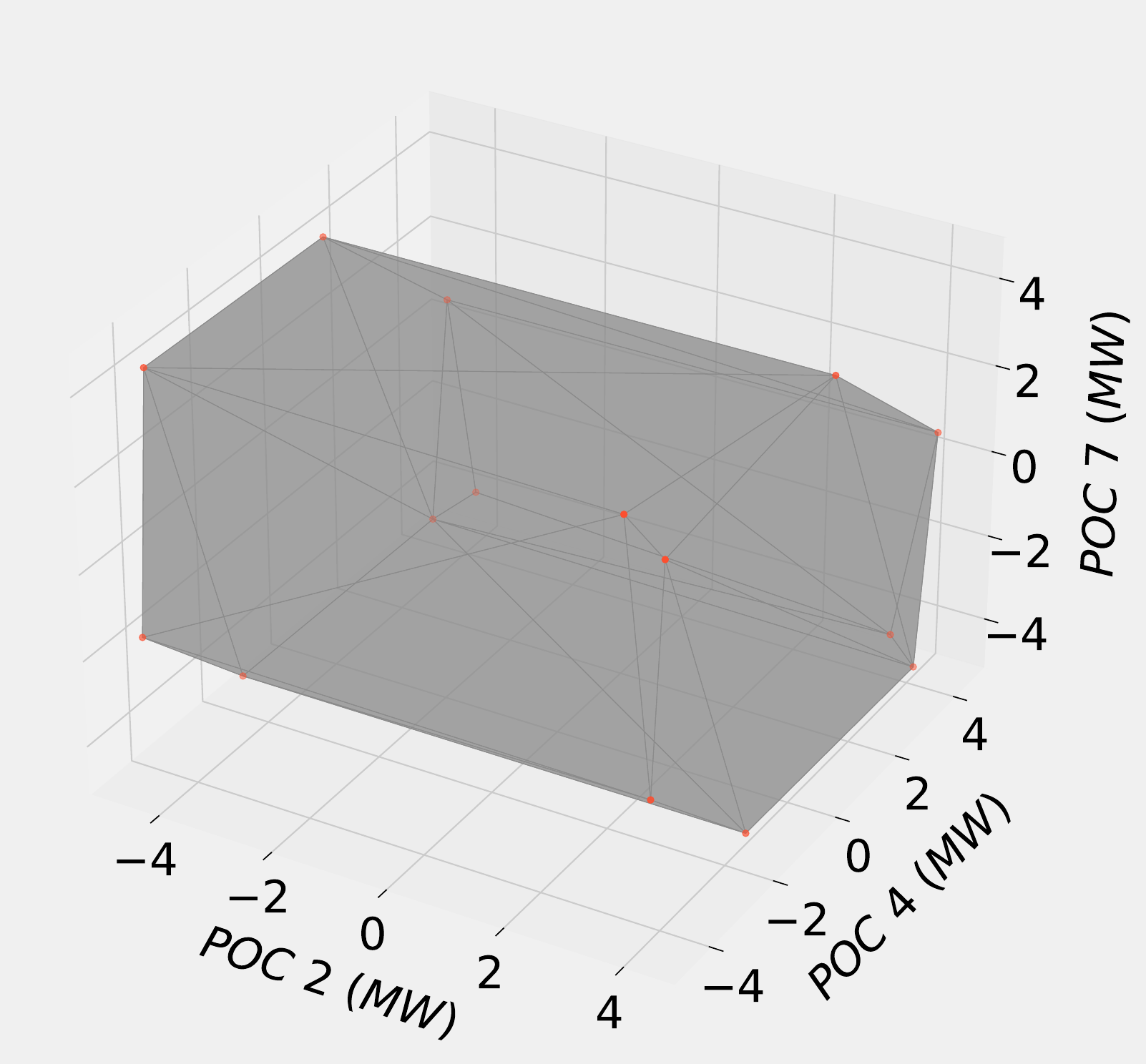}
         \caption{15-vertex SR of POC\,2,4,7}
         \label{fig:15SR247}
     \end{subfigure}
          \begin{subfigure}[b]{0.24\textwidth}
         \centering
         \includegraphics[width=\textwidth]{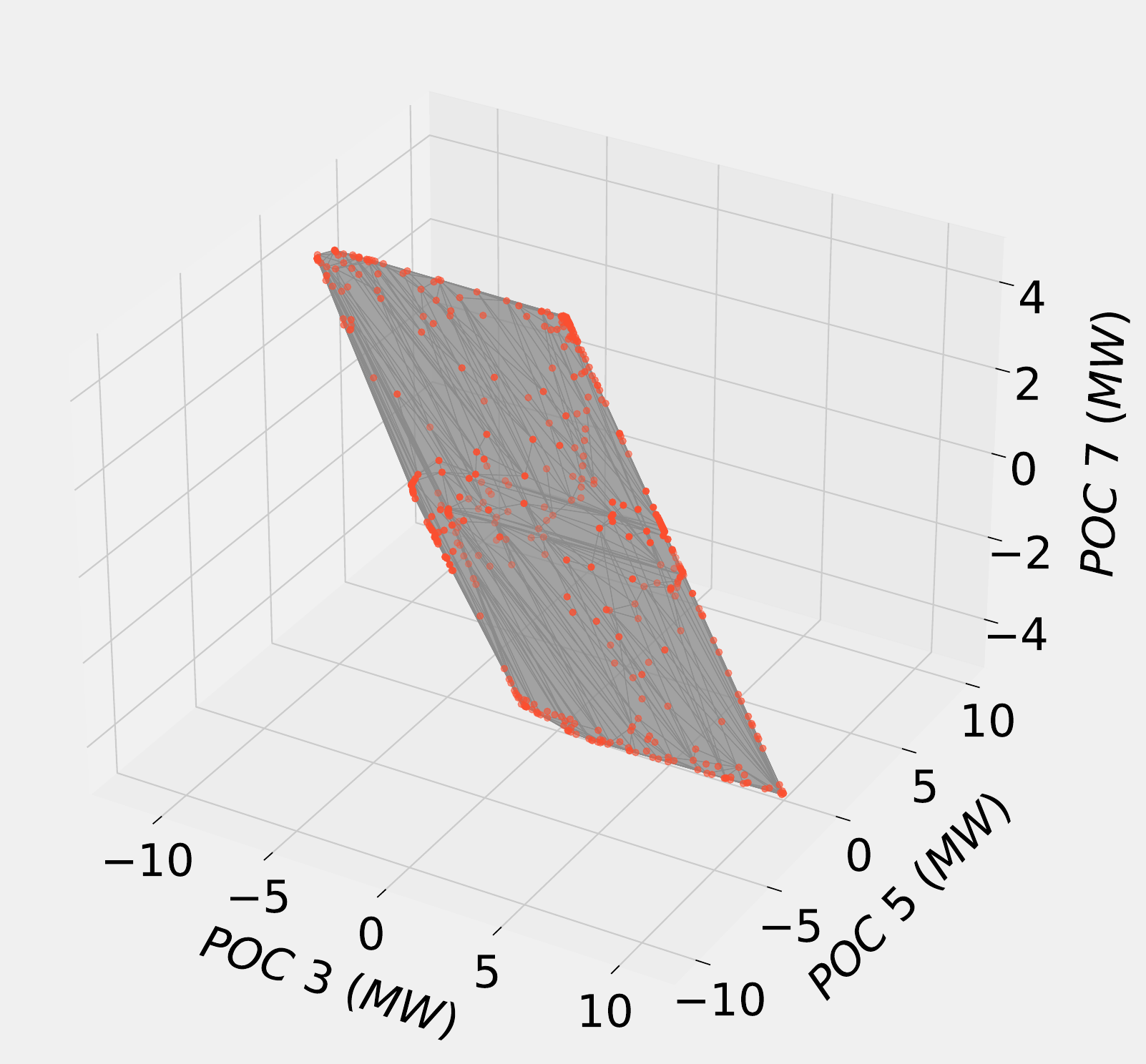}
         \caption{468-vertex  SR of POC\,3,5,7}
         \label{fig:477SR357}
     \end{subfigure}
     \begin{subfigure}[b]{0.24\textwidth}
         \centering
         \includegraphics[width=\textwidth]{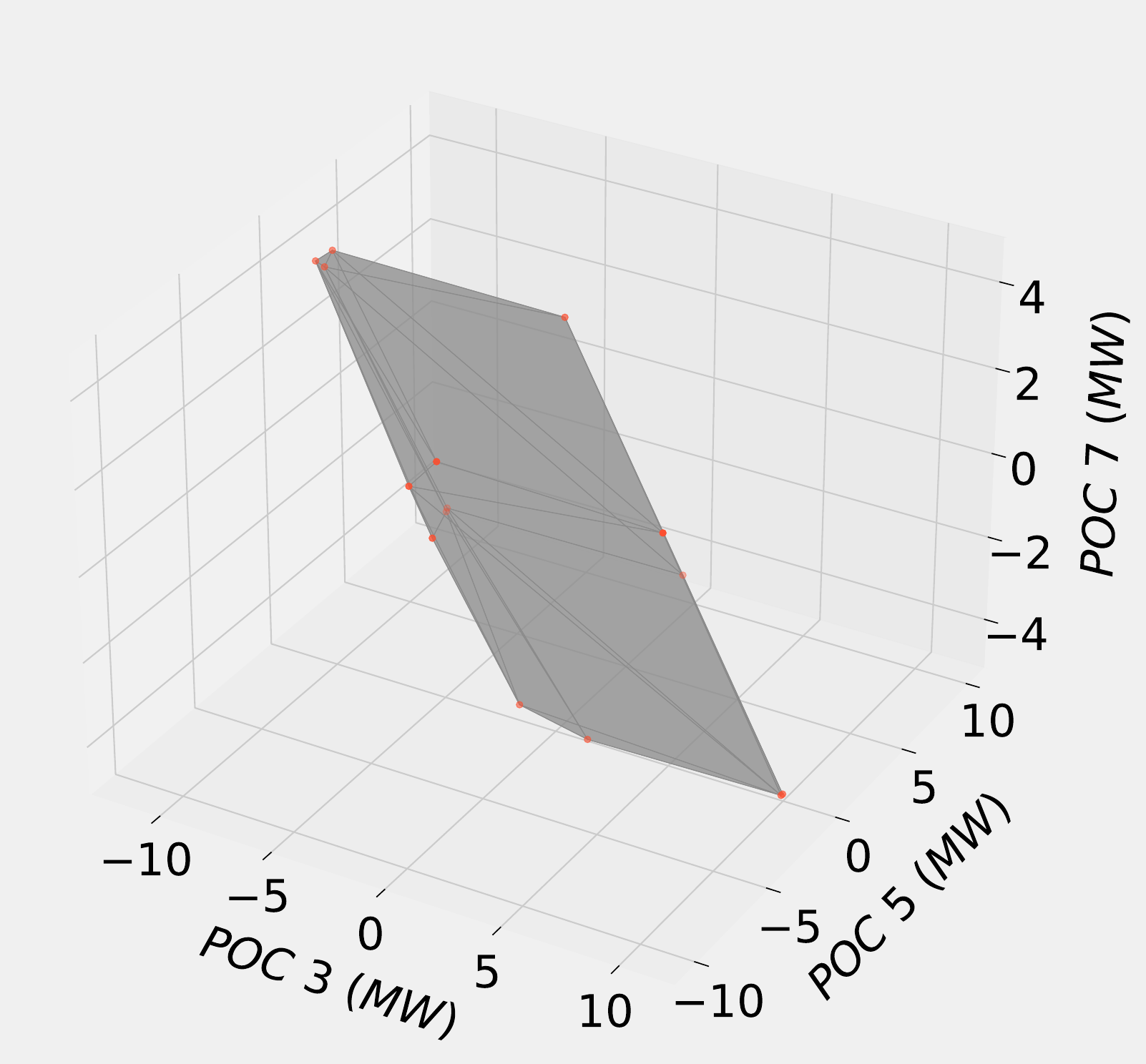}
         \caption{15-vertex SR of POC\,3,5,7}
         \label{fig:15SR357}
     \end{subfigure}
        \caption{The original and compressed SRs in the 9-bus system}
        \label{fig:compareSR}
\end{figure}

\paragraph{The second phase of hybrid vertex pruning}
Considering such advanced compression performance as above, in practice, the stopping criterion should work to interrupt the compression. Meanwhile, in a multidimensional case, the pure local search policy can be insufficient. To provide a complete view of the efficacy of this hybrid strategy, we further compress the 15-vertex SRs into 5-vertex ones. Since the hybrid approach requires the integration of the backward search policy, its neighborhood size is set to 50 before running the experiment. To validate the necessity of this backward search policy, we conduct a comparative analysis in which the backward policy is either adopted or not in this remaining compression task.

As depicted in Fig.~\ref{fig:15vertex}, compared to the pure local search policy, the hybrid pruning approach consistently contributes to a lower hypervolume loss.
Looking into the entire compression process, regarding the instance of POC\,3,5,7, compared to the pure local search policy, we preserve an extra 35.57\% of the original SR hypervolume by using this hybrid pruning approach. To this end, the efficacy of the proposed hybrid vertex pruning approach has been confirmed.

\begin{figure}[]
     \centering
     \begin{subfigure}[b]{0.24\textwidth}
         \centering
         \includegraphics[width=\textwidth]{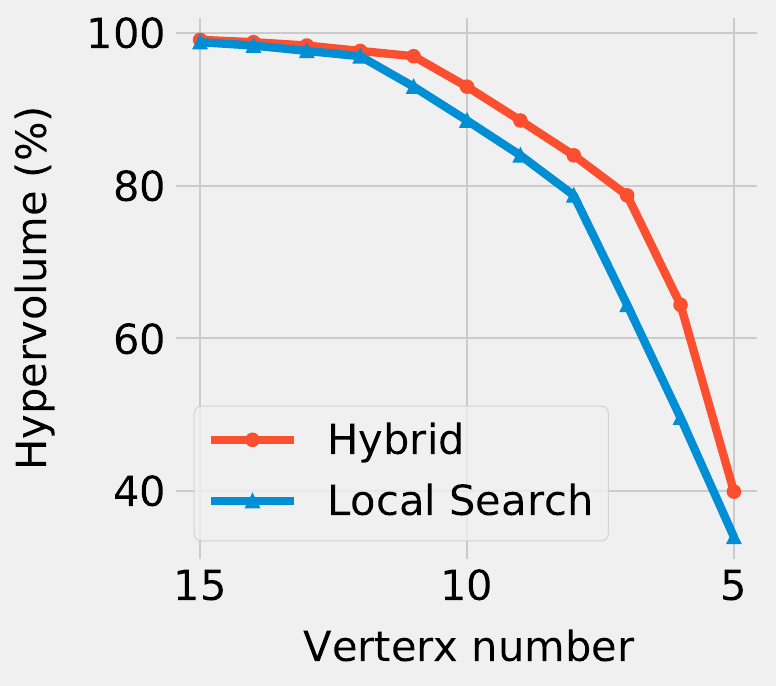}
         \caption{POC\,2,4,7}
         \label{fig:15vertex_247}
     \end{subfigure}
     \begin{subfigure}[b]{0.24\textwidth}
         \centering
         \includegraphics[width=\textwidth]{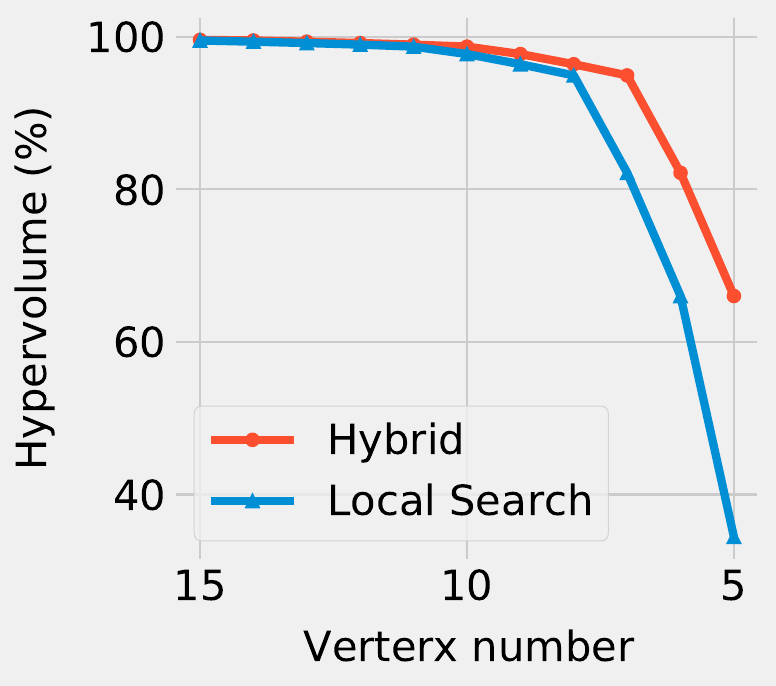}
         \caption{POC\,3,5,7}
         \label{fig:15vertex_357}
     \end{subfigure}
        \caption{The SR hypervolume curves in the 9-bus system}
        \label{fig:15vertex}
\end{figure}

\subsubsection{Results in the 296-bus testing system}
Similar results are observed in this case after employing Algorithm\,\ref{Ag:hybrid} for SR compression. The original and compressed SRs in both instances have been plotted in Fig.\,\ref{fig:compareSR_296}, exhibiting the visually negligible hypervolume loss after SR compression. The data size can be reduced by 96.0\% with 1.65\% hypervolume loss. Although the compression method performance in this case is weaker than that in the 9-bus system, the compression results remain satisfactory for the DSO. To this end, the effectiveness of  Algorithm\,\ref{Ag:hybrid} has been twice confirmed.
\begin{figure}[]
     \centering
     \begin{subfigure}[b]{0.24\textwidth}
         \centering
         \includegraphics[width=\textwidth, trim={0 0 0 2cm}, clip]{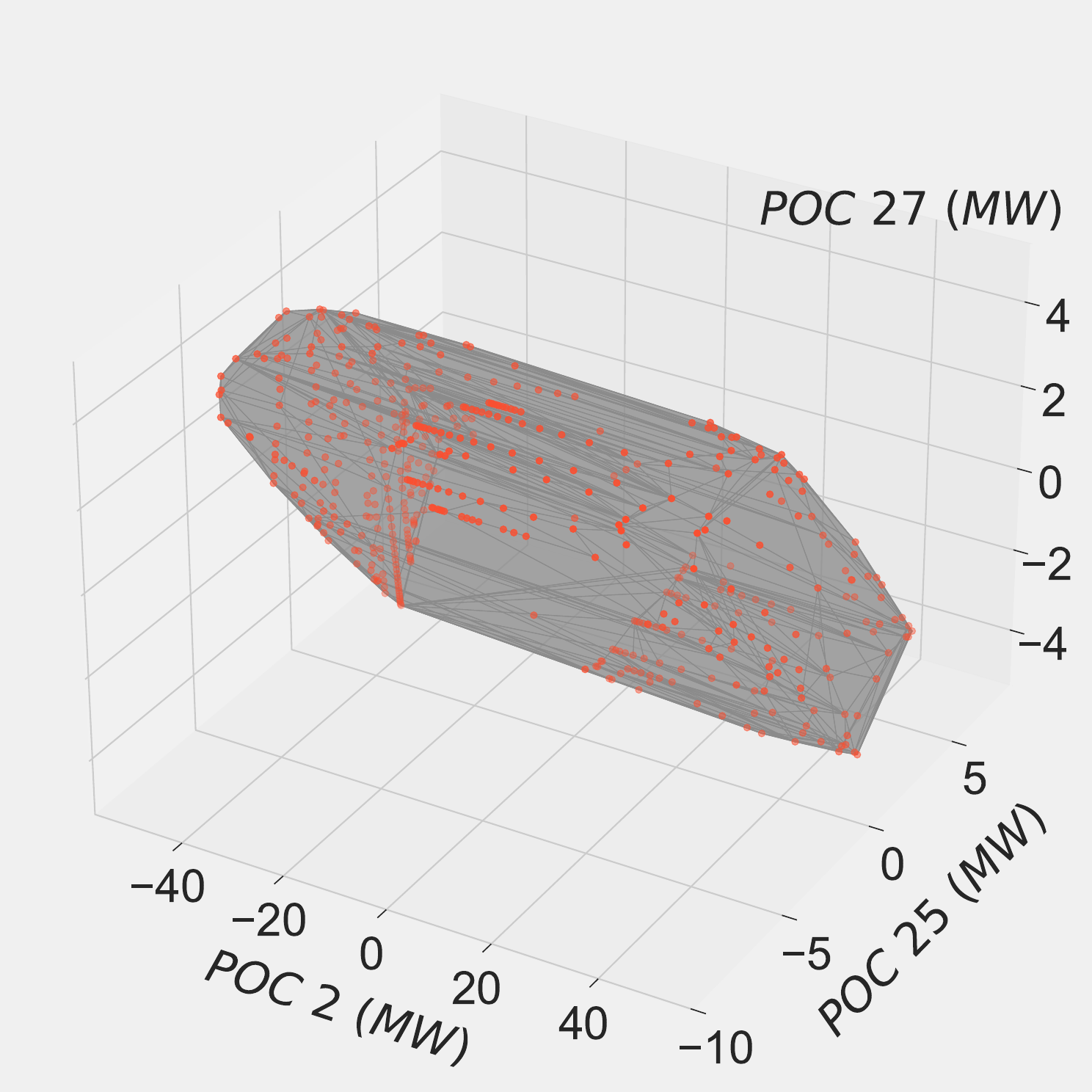}
         \caption{504-vertex SR of POC\,2,25,27}
         \label{fig:504SR22527}
     \end{subfigure}
     \begin{subfigure}[b]{0.24\textwidth}
         \centering
         \includegraphics[width=\textwidth, trim={0 0 0 2cm}, clip]{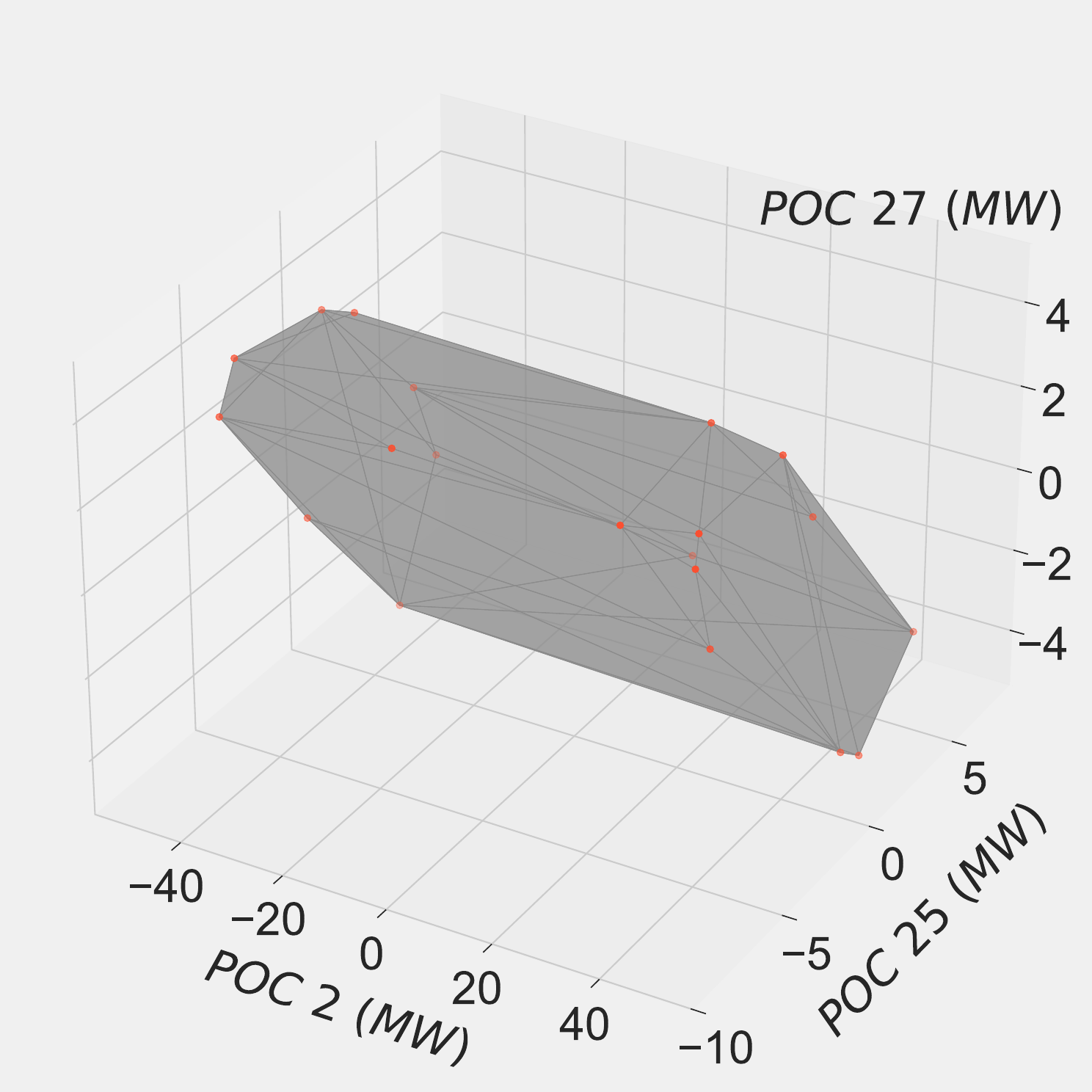}
         \caption{20-vertex SR of POC\,2,25,27}
         \label{fig:15SR22527}
     \end{subfigure}
        \caption{The original and compressed SRs in the 296-bus system}
        \label{fig:compareSR_296}
\end{figure}

\subsection{The SR dimensionality under limited memory space}
After validating the effectiveness of the proposed SR compression method in the 3-dimensional (3D) setting, this subsection applies Algorithm\,\ref{Ag:dimension} to compute the 8D SR. For the 9-bus system case study, the active power injections at POC\,1-8 define the eight dimensions of this SR. With the memory limit set to 500\,MB, $M$ is assigned an excessively large value so that the compression step in Algorithm\,\ref{Ag:dimension} is effectively disabled. Accordingly, this memory-careless SR computation method is executed until a memory overflow occurs. It is found that when the SR expands to a polytope with 158 vertices, the memory limit has been violated. In this situation, the computation process should be suspended, and this 158-vertex polytope will be taken as the result, and the corresponding SR hypervolume is only $1.14\times 10^{10}\,MW^8$. In parallel, under the same experimental conditions, the memory-efficient method as stated in Algorithm\,\ref{Ag:dimension} has been implemented. The upper bound on the number of vertices $M$ is set to 150, while the lower bound $\hat{M}$ is set to 75. The number of compression–expansion iterations is set to 10. The resulting SR hypervolume reaches \(9.89\times 10^{10}\,\text{MW}^8\), which is 7.67 times larger than the memory-careless result.

All computational outcomes of various computation methods are depicted in Fig.\,\ref{fig:vol_curve}, including the interim result in each iteration for memory-efficient calculation. Regarding the memory-careless computation method, its last available result prior to memory overflow exhibits a substantial computation error, as indicated by the significant hypervolume discrepancy relative to the SR-compression-enhanced result. This confirms that the memory-careless version fails to effectively complete the 8D SR calculation, although it has been proved able to compute the 3D SR efficiently in \cite{gong2023Multi}. To this end, with the memory-efficient advantage, SR compression is confirmed to be able to enhance the dimensionality of SR computation.
\begin{figure}[]
     \centering
         \centering
         \includegraphics[width=0.95\linewidth]{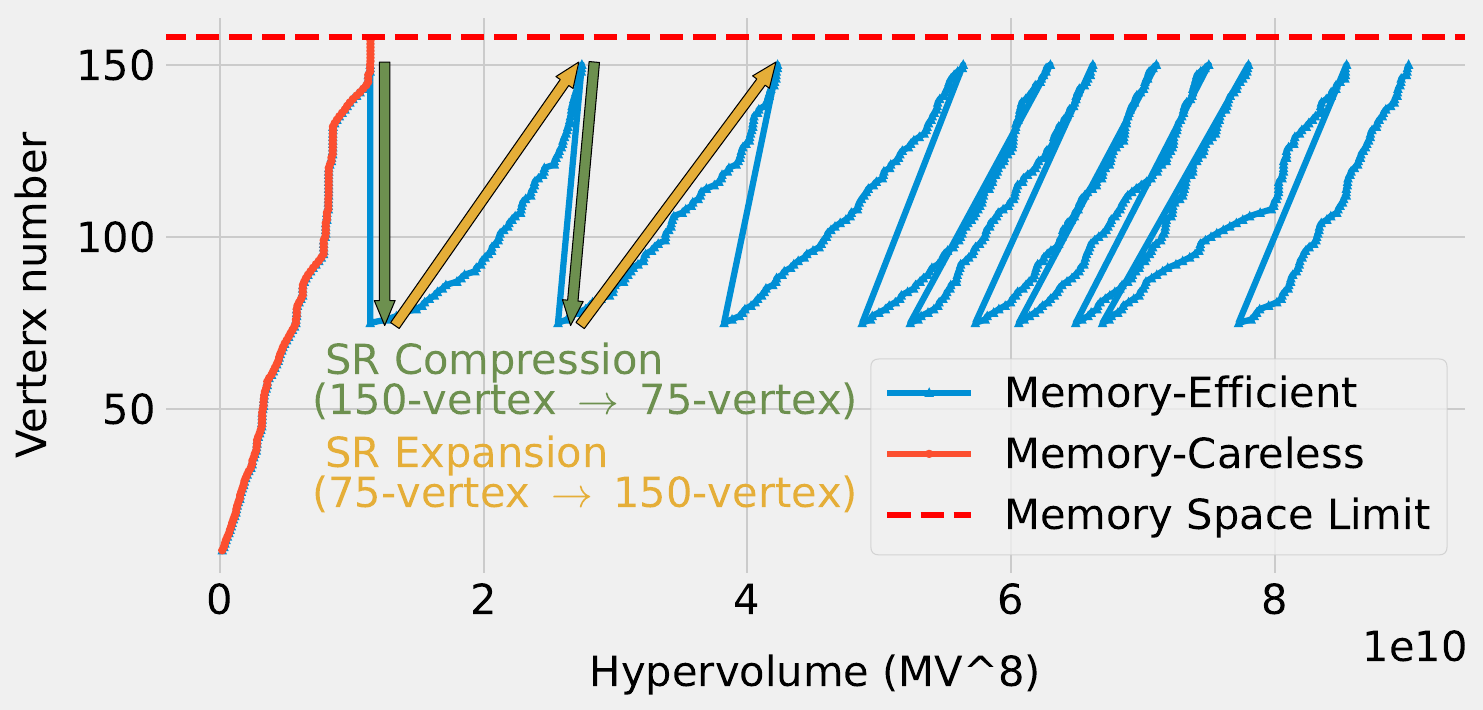}
        \caption{The hypervolume curve in 8D SR computation}
        \label{fig:vol_curve}
\end{figure}

\subsection{The applicability to grid-aware computation challenges}
\subsubsection{Optimal energy storage sizing}
\label{subsubsec:sizing}
To verify that the compressed SR compression may boost the solving process of grid-aware computation challenge. The general problem setting has been given in Section\,\ref{subsubsec:gridaware}. Focusing on the 9-bus network, we generate 720 testing scenarios by referring to \cite{gong2023Multi}. In each scenario, the power flow within the grid is infeasible, and we aim to minimize the cost of the energy storage equipment, which are additionally installed at relevant POCs to retrieve the grid feasibility. Through alternatively setting $\textbf{R}$ in (\ref{eq:region_op}) to the original 477-vertex or 468-vertex, the compressed 50-vertex, or the compressed 15-vertex SR, we compare the respective solving results to analyze the advantages of SR compression. The GUROBI 10.0 software serves as a benchmark solver to acquire the optimal solutions of (\ref{eq:region_op}) in all scenarios\,\cite{gurobi}.

By averaging the results of all testing scenarios, the corresponding computational performances under various SRs are listed at the left portion of Table\,\ref{tab:result}.
After replacing the original SR by the compressed 50-vertex one, considering the ignorable average objective value loss, this compressed SR helps to save over 64.4\% solving time in this grid-aware computational challenge. Regarding the further compressed 15-vertex one, the relative objective loss rises to 1.36\%, which is still acceptable for the grid operator in practice. Reversely, with this accuracy sacrifice, the respective solving time has been reduced by 87.1\%. Even with the SR computation time included, the summative computation time for all scenarios can be reduced by at least 52.9\%.

\begin{table*}[t]
    \centering
    \caption{Solution results in SR-related optimization models with various SRs}
    \begin{tabular}{ c | c c | c c c | c c c}
        \toprule
        \multirow{2}{*}{POCs} & \multirow{2}{*}{SR type} &\multirow{2}{*}{$T_c$\,(s)} & \multicolumn{3}{c|}{ Storage sizing (720 scenarios)} & \multicolumn{3}{c}{Capacity allocation (1000 scenarios)}\\
         & & & $O$\,(MW) & $T_a$\,(ms) & $T_s$\,(s) & $O$\,(MW) & $T_a$\,(ms) & $T_s$\,(s)\\
        \midrule
        \multirow{3}{*}{2,4,7}   & 477-vertex & - & 2.122 & 92.98 & 66.95 & 24.554 & 117.48 & 117.48\\
         & 50-vertex & 22.78 & 2.123 & 33.10 & 46.61&  24.545 & 38.28 & 61.06\\
                & 15-vertex & 22.86 & 2.151 & 12.04 & 31.53 & 24.500 & 12.05 & 34.91\\

        \midrule
        \multirow{3}{*}{3,5,7}   & 468-vertex & - & 5.592 & 115.72 & 83.32 & 14.010 & 168.83 & 168.83\\
         & 50-vertex & 26.05 & 5.592 & 32.74 & 49.62&  14.010 & 42.72 & 68.77\\
                & 15-vertex & 26.64 & 5.570 & 11.30 & 34.78 & 14.010 & 12.68 & 39.32\\
         \bottomrule
    \end{tabular}
    \begin{tablenotes}
      \item (1) $T_c$ represents the time required for SR compression. $O$ denotes the objective value. $T_a$ indicates the average solution time for solving the optimization model in each scenario. $T_s$ measures the total time spent on all scenarios, including both the compression time and the solution time.
\end{tablenotes}
    \label{tab:result}
\end{table*}

To visualize the objective value divergence of all 720 scenarios in the instance of POC\,2,4,7, particularly between the original SR and the compressed 15-vertex SRs, we draw their objective value distributions in Fig.~\ref{fig:error_compare}. Respective kernel distribution estimation plots are also provided, whose overlapping relationship further claims the high similarity of the optimization results\,\cite{hardle2004nonparametric}. Through calculating the absolute values of those deviations, the deviation distribution is also depicted in Fig.~\ref{fig:error_compare}, where most deviations are observed to be located in the range of 0-100\,kW and relatively low. Based on the above observations, in this store sizing problem, with the acceptable optimality loss, the effectiveness of SR compression on solution acceleration has been confirmed.
\begin{figure}
    \centering
    \includegraphics[width=\linewidth]{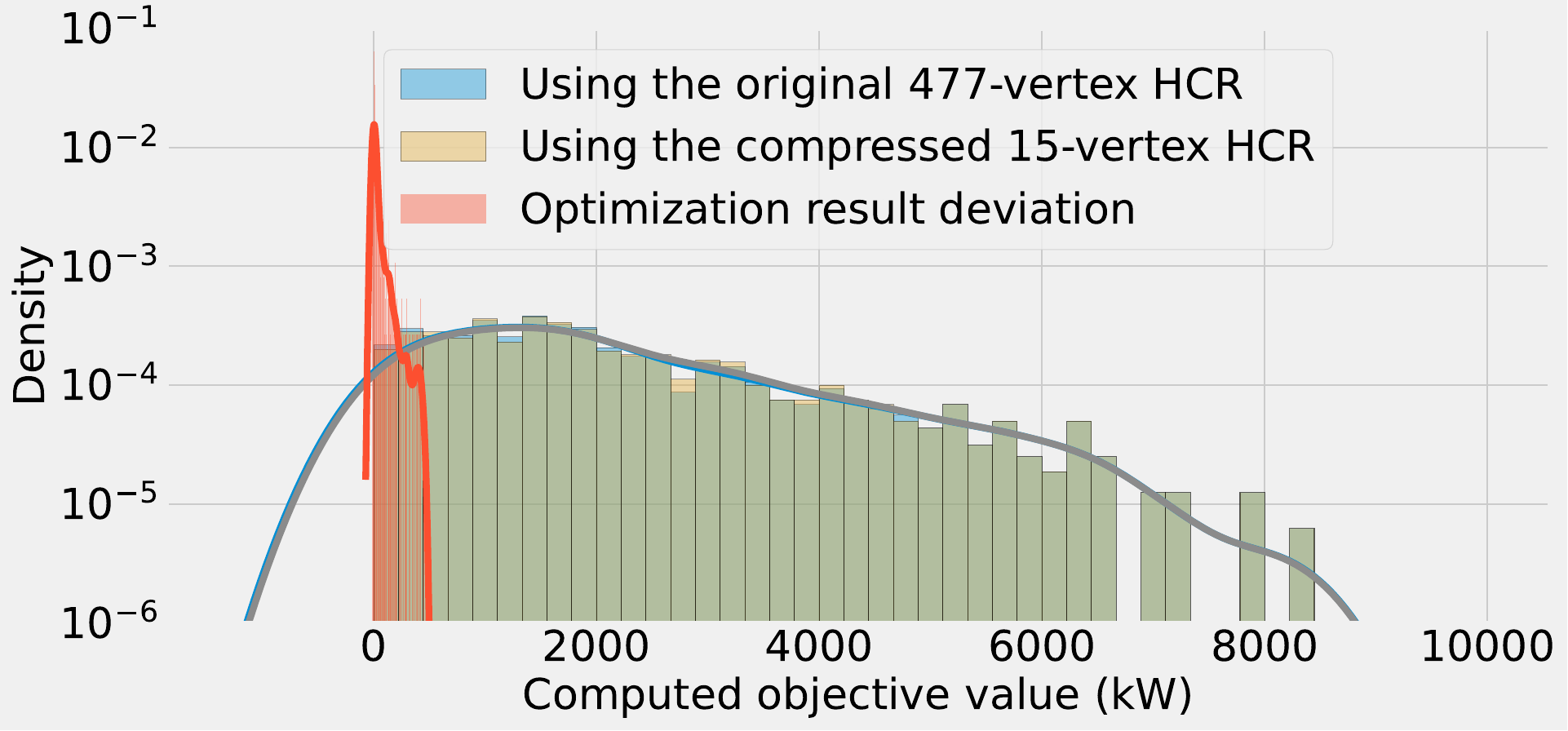}
    \caption{Result distributions in optimal energy storage sizing}
    \label{fig:error_compare}
\end{figure}

\subsubsection{Optimal capacity allocation}
Inheriting the testing environment in the last part, we aim to maximize the total capacity of all involved POCs in each scenario. The loading ratio of the energy unit at each POC is dynamic, indicating different $\alpha_j$ values in each scenario. This paper utilizes realistic measurement data of loading ratios over three EV chargers in the Netherlands, which are provided by ElaadNL\,\cite{ElaadData}. The plots of 5-minute-resolution charging profiles can be found at Chapter\,4.2 in \cite{gong2025thesis}. Based on this dataset, there are 1000 scenarios being sampled for integration as different $\alpha_i$ values. $\Hat{\alpha}_i$ is set to be 80\%.

Sharing the same computation environment with that in Section\,\ref{subsubsec:sizing}, we still compare the results under original and compressed SRs. By averaging the results of all 1000 testing scenarios, the corresponding computation performances are listed by the right portion of Table\,\ref{tab:result}. After replacing the original SR by the compressed 50-vertex one, considering the ignorable average objective value loss, this compressed SR helps to save over 67.4\% solving time in this grid-aware computational challenge. Regarding the further compressed 15-vertex one, the relative objective loss rises to 0.22\%, which is still acceptable for the grid operator in practice. Reversely, with this accuracy sacrifice, the respective solving time has been reduced by 89.7\%. By including the SR computation time, the summative computation time for all scenarios can be reduced by at least 70.2\%.

Furthermore, alike Fig.\,\ref{fig:error_compare}, Fig.\,\ref{fig:error_compare2} is provided to visualize the objective value divergence when employing different SRs in the instance of POC\,2,4,7. Along with the high similarity between the objective value distribution under various SRs, the close-to-zero deviation distribution has confirmed the limited  influence of optimization results after the SR compression. To this end,  in the capacity allocation problem, the benefits of SR compression on computational task acceleration with limited optimality loss have also been confirmed.

\begin{figure}
    \centering
    \includegraphics[width=\linewidth]{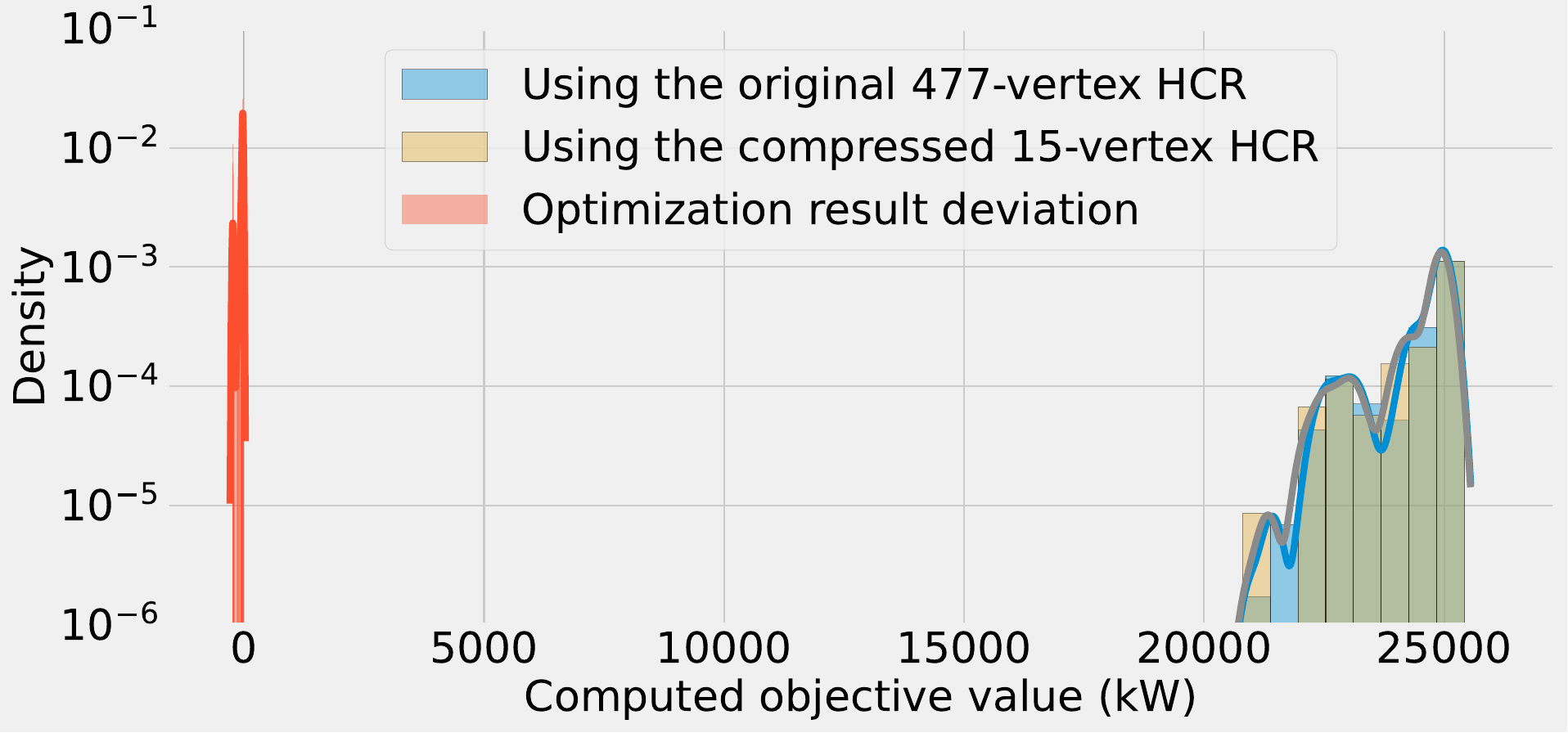}
    \caption{Result distributions in optimal capacity allocation}
    \label{fig:error_compare2}
\end{figure}

\subsection{Summary}
Using real-world grid cases in various scales, the efficacy of the proposed SR compression approach has been experimentally validated. Furthermore, both benefits of SR compression, including enhanced SR dimensionality and accelerated grid-aware optimization model solving, have also been confirmed to clarify the value of the proposed method in this paper.

\section{Conclusion}
\label{sec:conclusion}
The SR may bring significant value to contemporary grid operation practices, including the exploitation of grid power delivery potential and the acceleration of grid-aware computation tasks. Concerning the first part, its applicability is influenced by ``curse of dimensionality". Regarding the second part, the value amount is affected by the SR data size. Both motivate the implementation of SR compression. After clarifying the mathematical equivalence between the optimal SR compression and vertex pruning, this paper proposes a hybrid vertex pruning approach that mixes the local and backward search policies, and integrates the VNS method to achieve a balance between computational complexity and optimality. The benefit of this method in enhancing computational efficiency is evidenced by its role in enabling memory-efficient high-dimensional SR computations and in providing time-efficient solutions to SR-related optimization problems.

Relevant numerical studies have verified the effectiveness of this method, with which we have achieved over 96.0\% data size reduction with at most 1.65\% hypervolume loss. As a byproduct, the benefits of SR compression in improving the SR dimensionality and accelerating the grid-aware computation process have also been validated. Looking ahead, more interdisciplinary methods, such as reinforcement learning or random geometry analysis, will be investigated for pruning policy integration in future works. We aim to uncover their potential to support more high-efficiency compression strategies, particularly in high-dimensional SR scenarios.

\appendices

\section{An counterexample of long-term
optimality with the local search policy}
\label{app:counter}
To prove this long-term optimality loss of the local search policy, a 2D counterexample is given, where the 5-vertex polygon is anticipated to be compressed into a 3-vertex one. As illustrated in Fig.\,\ref{fig:counter}, By following the local research policy, the hypervolume loss will be visualized as Regions\,1 and 2 in two steps, leaving Region\,3 as the compressed one. Meanwhile, by solving the optimization problem (\ref{eq:prime}), the orange triangle in Fig.\,\ref{fig:good} is derived, which has a larger space than that of Region\,3. This result confirms the loss of long-term optimality attributed to the local search policy.

\begin{figure}[]
     \centering
      \begin{subfigure}[b]{0.24\textwidth}
         \centering
         \includegraphics[width=\textwidth]{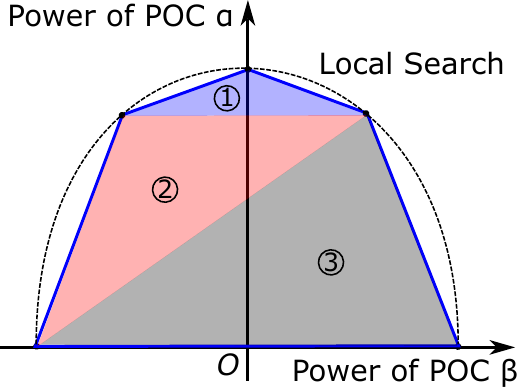}
         \caption{Local search result}
         \label{fig:counter}
     \end{subfigure}
     \begin{subfigure}[b]{0.24\textwidth}
         \centering
         \includegraphics[width=\textwidth]{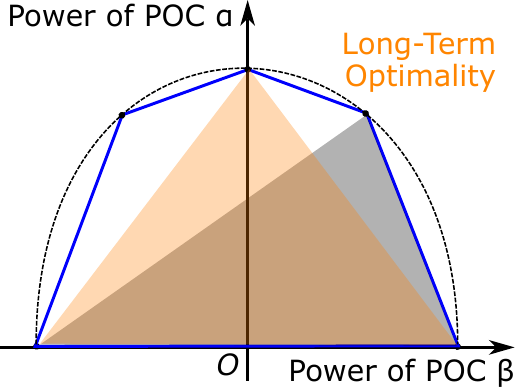}
         \caption{Long-term optimal result}
         \label{fig:good}
     \end{subfigure}
        \caption{The long-term optimality loss of local search policy}
        \label{fig:counterexample}
\end{figure}

\section*{Acknowledgment}
The authors would like to acknowledge the grid data reference from Alliander N.V. and Enexis Groep, which serve as the Dutch DSOs and the coordinators for the NEON project. ElaadNL is also appreciated for providing charging data within the NEON framework.
\ifCLASSOPTIONcaptionsoff
  \newpage
\fi

\nocite{*}
\bibliographystyle{IEEEtran}
\bibliography{reference}

\end{document}